\documentclass[11pt]{article}
\usepackage{fullpage}
\usepackage{paralist}
\usepackage{times}
\usepackage{amsmath,amsfonts,amssymb,amsthm}
\usepackage[colorlinks=true,linkcolor=red,citecolor=blue,urlcolor=red]{hyperref}
\usepackage{framed}
\usepackage{verbatim}
\usepackage[T1]{fontenc}
\usepackage{paralist}

\usepackage{mathptmx}
\usepackage[text={6.5in,8.8in}]{geometry}

\usepackage[compact]{titlesec}
\titlespacing{\section}{0pt}{3ex}{2ex}
\titlespacing{\subsection}{0pt}{2ex}{1ex}
\titlespacing{\subsubsection}{0pt}{0.5ex}{0ex}

\begin{document}

\newtheorem{fact}{Fact}[section]
\newtheorem{rules}{Rule}[section]
\newtheorem{conjecture}{Conjecture}[section]
\newtheorem{theorem}{Theorem}[section]
\newtheorem{hypothesis}{Hypothesis}
\newtheorem{problem}{Problem}
\newtheorem{remark}{Remark}
\newtheorem{proposition}{Proposition}
\newtheorem{corollary}{Corollary}[section]
\newtheorem{lemma}{Lemma}[section]
\newtheorem{claim}{Claim}
\newtheorem{definition}{Definition}[section]

\newenvironment{proofof}[1]{\medskip
\noindent {\bf Proof of #1.  }}{\hfill$\Box$
\medskip}

\newenvironment{reminder}[1]{\medskip
\noindent {\bf Reminder of #1  }\em}{
\smallskip}

\def \iouseful {\text{useful}}
\def \BP {{\sf BP}}
\def \coRP {{\sf coRP}}
\def \EXACT {{\sf EXACT}}
\def \SYM {{\sf SYM}}
\def \SAC {{\sf SAC}}
\def \SUBEXP {{\sf SUBEXP}}
\def \ZPSUBEXP {{\sf ZPSUBEXP}}
\def \SYMACC {{\sf SYM\text{-}ACC} }
\def \QED {{\hfill$\Box$}}
\def \PH {{\sf PH}}
\def \RP {{\sf RP}}
\def \coMA {{\sf coMA}}
\def \ZPTISP {{\sf ZPTISP}}
\def \REXP {{\sf REXP}}
\def \coNP {{\sf coNP}}
\def \BPP {{\sf BPP}}
\def \NC {{\sf NC}}
\def \ZPE {{\sf ZPE}}
\def \NE {{\sf NE}}
\def \E {{\sf E}}
\def \poly { \text{\rm poly} }
\def \TC {{\sf TC}}
\def \DTS {{\sf DTS}}
\def \R {{\mathbb R}}
\def \Z {{\mathbb Z}}
\def \P {{\sf P}}
\def \MA {{\sf MA}}
\def \AM {{\sf AM}}
\def \MATIME {{\sf MATIME}}
\def \QP {{\sf QP}}
\def \coNQP {{\sf coNQP}}
\def \NP {{\sf NP}}
\def \EXP {{\sf EXP}}
\def \NTISP {{\sf NTISP}}
\def \DTISP {{\sf DTISP}}
\def \TISP {{\sf TISP}}
\def \T {{\sf TIME}}
\def \TIME {\T}
\def \Sig[#1] {{\sf \Sigma}_{#1} }
\def \Pie[#1] {{\sf \Pi}_{#1} }
\def \NTS {{\sf NTS}}
\def \NSP {{\sf NSPACE}}
\def \NSPACE {{\sf NSPACE}}
\def \BPACC {{\sf BPACC}}
\def \ACC {{\sf ACC}}
\def \ASP {{\sf ASPACE}}
\def \io {\textrm{\it io-}}
\def \ro {\textrm{\it ro-}}
\def \ATISP {{\sf ATISP}}
\def \SIZE {{\sf SIZE}}
\def \AT {{\sf ATIME}}
\def \AC {{\sf AC}}
\def \BPTIME {{\sf BPTIME}}
\def \SPACE {{\sf SPACE}}
\def \RE {{\sf RE}}
\def \ZPSUBEXP {{\sf ZPSUBEXP}}
\def \coREXP {{\sf coREXP}}
\def \NQL {{\sf NQL}}
\def \QL {{\sf QL}}
\def \RTIME {{\sf RTIME}}
\def \NSUBEXP {{\sf NSUBEXP}}
\def \MAEXP {{\sf MAEXP}}
\def \PP {{\sf PP}}
\def \PSPACE {{\sf PSPACE}}
\def \NT {{\sf NTIME}}

\def \NTIME {\NT}

\def \ATIME {\AT}

\def \Z {{\mathbb Z}}
\def \F {{\mathbb F}}
\def \NTIBI {{\sf NTIBI}}

\def \TIBI {{\sf TIBI}}

\def \QBFk {{\text{QBF}_k}}

\def \coNT {{\sf coNTIME}}

\def \coNTISP {{\sf coNTISP}}

\def \coNTIME {\coNT}

\def \coNTIBI {{\sf coNTIBI}}
\def \MOD {{\sf MOD}}
\def \ZPEXP {{\sf ZPEXP}}
\def \ZPTIME {{\sf ZPTIME}}
\def \ZPP  {{\sf ZPP}}
\def \DT {{\sf DTIME}}
\def \coNE {{\sf coNE}}
\def \DTIME {\DT}

\def \isin {\subseteq}

\def \isnotin {\nsubseteq}

\def \L {{\sf LOGSPACE}}

\def \LOGSPACE {\L}

\def \N {{\mathbb N}}
\def \NQP {{\sf NQP}}
\def \NEXP {{\sf NEXP}}

\def \coNEXP {{\sf coNEXP}}

\def \SIZE {{\sf SIZE}}

\def \eps {\varepsilon}

\newcommand{\card}[1]{\ensuremath{\left|#1\right|}}
\newcommand{\ip}[2]{\ensuremath{\left<#1,#2\right>}}
\newcommand{\mv}[2]{\ensuremath{\mathbf{MV}\!\left(#1,#2\right)}}

\renewcommand{\lg}{\log}

\title{Faster Online Matrix-Vector Multiplication}
\author{Kasper Green Larsen\footnote{Department of Computer Science, Aarhus University, \texttt{larsen@cs.au.dk}. Supported by Center for Massive Data Algorithmics, a Center of the Danish National Research Foundation, grant DNRF84, a Villum Young Investigator Grant and an AUFF Starting Grant.} \and 
Ryan Williams\footnote{Computer Science Department, Stanford University, {\tt rrw@cs.stanford.edu}. Supported in part by NSF CCF-1212372 and CCF-1552651 (CAREER). Any opinions, findings and conclusions or recommendations expressed in this material are those of the authors and do not necessarily reflect the views of the National Science Foundation.}}

\date{}

\maketitle

\begin{abstract} We consider the \emph{Online Boolean Matrix-Vector Multiplication} (OMV) problem studied by Henzinger \emph{et al.} [STOC'15]: given an $n \times n$ Boolean matrix $M$, we receive $n$ Boolean vectors $v_1,\ldots,v_n$ one at a time, and are required to output $M v_i$ (over the Boolean semiring) before seeing the vector $v_{i+1}$, for all $i$. Previous known algorithms for this problem are combinatorial, running in $O(n^3/\log^2 n)$ time. Henzinger \emph{et al.} conjecture there is no $O(n^{3-\eps})$ time algorithm for OMV, for all $\eps > 0$; their OMV conjecture is shown to imply strong hardness results for many basic dynamic problems.

We give a substantially faster method for computing OMV, running in $n^3/2^{\Omega(\sqrt{\log n})}$ randomized time. In fact, after seeing $2^{\omega(\sqrt{\log n})}$ vectors, we already achieve $n^2/2^{\Omega(\sqrt{\log n})}$ amortized time for matrix-vector multiplication. Our approach gives a way to reduce matrix-vector multiplication to solving a version of the Orthogonal Vectors problem, which in turn reduces to ``small'' algebraic matrix-matrix multiplication. Applications include faster independent set detection, partial match retrieval, and 2-CNF evaluation. 

We also show how a modification of our method gives a cell probe data structure for OMV with worst case $O(n^{7/4}/\sqrt{w})$ time per query vector, where $w$ is the word size. This result rules out an unconditional proof of the OMV conjecture using purely information-theoretic arguments. 
\end{abstract}

\thispagestyle{empty}
\newpage
\setcounter{page}{1}

\section{Introduction}

We consider the following problem defined over the Boolean semiring (with addition being OR, and multiplication being AND): 

\begin{problem}[Online Matrix-Vector Multiplication (OMV)] Given a matrix $M \in \{0,1\}^{n\times n}$, and a stream of vectors $v_1,\ldots,v_n \in \{0,1\}^n$,  the task is to output $Mv_i$ before seeing $v_{i+1}$, for all $i=1,\ldots,n-1$. 
\end{problem}

While OMV is a natural algebraic problem, its online requirement precludes the use of fast matrix multiplication algorithms, a la Strassen~\cite{Strassen69}. Indeed, no OMV algorithms are known which run much faster than $O(n^3)$ (modulo $\log^2 n$ factors~\cite{Williams07}). Henzinger \emph{et al.}~\cite{HenzingerKNS15} posed the conjecture:

\begin{conjecture}[OMV Conjecture] For every $\eps > 0$, there is no $O(n^{3-\eps})$-time randomized algorithm that solves OMV with an error probability of at most $1/3$.
\end{conjecture} 

Assuming the OMV Conjecture, the authors proved \emph{tight} time lower bounds for over a dozen basic dynamic problems, in the fully and partial dynamic setting. Thus the OMV problem is a simple yet ubiquitous barrier to improving the query times for many dynamic problems. Indeed, strong matrix-vector multiplication lower bounds are known. Recent work of Clifford, Gr{\o}nlund, and Larsen \cite{CliffordGL15} show that any $\poly(n)$-space data structure for matrix-vector multiplication over sufficiently large fields requires $\Omega(n^2/\log n)$ time. In another direction, a simple counting argument (cf.~\cite{Moon-Moser66}) shows that there are $n\times n$ matrices $M$ over $\F_2$ such that every circuit computing $M$ (the linear transformation corresponding to $M$) requires $\Omega(n^2/\log n)$ gates. Such a counting argument can easily be extended to give the same lower bound for computing $M$ over the Boolean semiring.

\subsection{Our Results}

\paragraph{Faster online matrix-vector multiplication.} Via an entirely new approach, we show how to compute OMV in $o(n^2/(\log n)^c)$ amortized time per query $v_i$, for every constant $c > 0$:

\begin{theorem} \label{OMV} With {\bf no} preprocessing of the matrix $A \in \{0,1\}^{n \times n}$, and for any sequence of $t=2^{\omega(\sqrt{\log n})}$ vectors $v_1,\ldots,v_t \in \{0,1\}^n$, online matrix-vector multiplication of $A$ and $v_i$ over the Boolean semiring can be performed in $n^2/2^{\Omega(\sqrt{\log n})}$ amortized time, with a randomized algorithm that succeeds whp.\footnote{Note that in Theorem~\ref{OMV} and throughout the paper, ``with high probability'' (whp) means probability at least $1-1/\poly(n)$.}
 \end{theorem}

We stress that the amortized bound already takes effect after only $n^{o(1)}$ query vectors. 
The success of our algorithm critically relies on the Boolean nature of the OMV problem, and the amortized condition: by restricting to the Boolean setting, we avoid the $\Omega(n^2/\log n)$ lower bounds of \cite{CliffordGL15} for arbitrary fields; by keeping an amortized data structure, we avoid $\Omega(n^2/\log n)$ lower bounds that arise from counting arguments. More details of the algorithm are given in Section~\ref{intuition}.

The faster OMV algorithm allows one to quickly compute the neighborhood of any subset of vertices in a given graph. This gives rise to several immediate applications, such as:

\begin{corollary} \label{independent} For every graph $G = (V,E)$, after $O(n^2)$-time preprocessing there is a data structure such that, for every subset $S \subseteq V$, we can answer online whether $S$ is independent, dominating, or a vertex cover in $G$, in $n^2/2^{\Omega(\sqrt{\log n})}$ amortized time over $2^{\omega(\sqrt{\log n})}$ subset queries.
\end{corollary}

\paragraph{Partial Matches.} The OMV data structure is also useful for answering \emph{partial match retrieval queries} faster than the naive bound. Let $\Sigma = \{\sigma_1,\ldots,\sigma_k\}$ and let $\star$ be an element such that $\star \notin \Sigma$. In the partial match retrieval problem, we are given $x_1,\ldots,x_n \in \{\Sigma \cup \{\star \}\}^m$ with $m \leq n$, and wish to preprocess them to answer ``partial match'' queries $q \in (\Sigma \cup \{\star\})^m$, where we say $q$ matches $x_i$ if for all $j=1,\ldots,m$, $q[j] = x_i[j]$ whenever $q[j] \neq \star$ and $x_i[j] \neq \star$. A query answer is the vector $v \in \{0,1\}^n$ such that $v[i] = 1$ if and only if $q$ matches $x_i$. 

Note this is a different setup from the usual literature, where a query answer is permitted to be a single string that matches $q$. In our setup, the running time for queries is necessarily $\Omega(n)$, and the running time for a single query is naively $O(nm\log k)$ (omitting possible log-factors saved from the word-RAM). Using our OMV algorithm, we can beat this running time considerably in the case of long strings (large $m$):

\begin{theorem} 
\label{thm:partialmatch} Let $\Sigma = \{\sigma_1,\ldots,\sigma_k\}$ and let $\star$ be an element such that $\star \notin \Sigma$. For any set of strings $x_1,\ldots,x_n \in \{\Sigma \cup \{\star \}\}^m$ with $m \leq n$, after $\tilde{O}(nm)$-time preprocessing, there is an $O(nm)$ space data structure such that, for every query string $q \in (\Sigma \cup \{\star\})^m$, we can answer online whether $q$ matches $x_i$, for every $i=1,\ldots, n$, in $(nm \log k)/2^{\Omega(\sqrt{\log m})}$ amortized time over $2^{\omega(\sqrt{\log m})}$ queries.
\end{theorem}

In previous work by Charikar, Indyk and Panigrahy~\cite{charikar:match}, they presented a partial match data structure that for any $c \leq n$ can be implemented with space usage $O(nm^c)$ and query time $O(nm/c)$. Their data structure is for the \emph{detection} variant where we must only determine whether there is an $x_j$ matching the query $q$. To achieve a query time matching ours, they must set $c=2^{\Omega({\sqrt{\lg m}})}/\log k$, resulting in a large space usage of $n m^{2^{\Omega(\sqrt{\lg m})}/\log k}$. For $m$ being just slightly super-polylogarithmic in $k$ ($m = 2^{\omega((\log \log k)^2)}$), this space bound simplifies to $n m^{2^{\Omega(\sqrt{\lg m})}}$. (Compare this to our space usage of $O(nm)$.) They also presented an alternative data structure when the desired query time is $o(n)$, which answers queries in time $O(n/c)$ and uses space $n \cdot 2^{O(m \lg^2 m \sqrt{c / \lg n})}$. 

A batched version of partial match has also been considered before. Here, one is given $n$ queries $q_1,\dots,q_n$ simultaneously and must report for each $q_i$ whether there is an $x_j$ that matches $q_i$. Recent work by Abboud, Williams and Yu~\cite{abboud:batched} shows how to solve the Boolean version ($\Sigma = \{0,1\}$) of this problem in time $O(n^{2-1/\log(m/\log n)})$ for $m = \Omega(\log n)$.

More applications can be found in Section~\ref{sec:applications}, including a faster online algorithm for vertex triangle queries, faster online evaluation of 2-CNF formulas on given assignments and a worst case efficient version of our OMV data structure.

\paragraph{A truly subquadratic-time cell probe data structure.}
Unlike other popular algorithmic hardness conjectures like SETH, the 3SUM conjecture, the APSP conjecture etc., the OMV conjecture asserts a lower bound on data structures instead of traditional algorithms. Given the state-of-the-art in proving data structure lower bounds compared to lower bounds for algorithms, there would seem to be more hope of giving an unconditional proof of the OMV conjecture in the not-so-distant future. For instance, the lower bound for matrix-vector multiplication over finite fields \cite{CliffordGL15} is $\Omega(n^2/\log n)$ for large enough fields. 

Lower bounds for word-RAM data structures are generally proved in the cell probe model. In this model, the running time of a data structure is measured solely by the number of memory cells that needs to be read/probed to answer a query; all computation is free-of-charge. These lower bound proofs are therefore purely \emph{information-theoretic}, typically showing that if data structures are too fast, then they have not "collected" enough information to answer queries. While ignoring computation time may allow for more efficient data structures, it turns out that one can most often prove cell probe lower bounds that match the best word-RAM upper bounds, see e.g.~\cite{patrascu06pred, Patrascu:loga, larsen:dynamic_count, patrascu08structures, larsen:median, marked, yu:cellprobe}.

Our second main theorem shows that the OMV conjecture is actually \emph{false} in the cell probe model. More specifically, we demonstrate how the OMV algorithm from Theorem~\ref{OMV} can be modified to give a cell probe data structure for OMV with \emph{truly subquadratic} worst-case time per query vector. Therefore, if the OMV conjecture is true, the bottleneck cannot be an information-theoretic one; it must be computational. 

\begin{theorem}
\label{OMVCellProbe}
There is a cell probe data structure that given a matrix $A \in \{0,1\}^{n \times n}$, preprocesses the matrix into an $O(n^2)$ space data structure, such that for every given query vector $v \in \{0,1\}^n$, the data structure computes $Av$ with $O(n^{7/4}/\sqrt{w})$ probes in the worst case (where $w$ is the word size).
\end{theorem}

Notice that $O(n^2)$ space is linear in the input size. In fact, our data structure is very succinct: it simply stores the matrix $A$, plus an additional $n^{7/4}\sqrt{w}$ bits. For the most typical word size of $w = \Theta(\lg n)$, this is only $\tilde{O}(n^{7/4})$ redundant bits. 

\subsection{Intuition}\label{intuition}

Let us outline the key ideas in the proof of Theorem~\ref{OMV}. First, a reduction of Henzinger \emph{et al.} (similar to earlier work of Vassilevska and Williams~\cite{Williams2010subcubic}) shows that to compute OMV in amortized $n^2/2^{\Omega(\sqrt{\log n})}$ time per query, it suffices to compute $u^T A v$ for arbitrary vectors $u,v \in \{0,1\}^n$ in $n^2/2^{\Omega(\sqrt{\log n})}$ amortized time per $(u,v)$ pair. 

To compute $u^T A v$, we reason as follows. The vectors $u$ and $v$ define some submatrix $A'$ of $A$, and our task is to determine if $A'$ is all-zeroes. We break the task into several cases:
\begin{enumerate}
\item If $A'$ contains many ones, then this is easy to determine by random sampling. 
\item If $A'$ contains few ones, and $A'$ is large, then we run in $O(n^2)$ time and ``extract'' $A'$ from $A$: we save the query $(u,v)$ in a list $L$, along with a sparse list of all the $1$-entries in the submatrix $A'$, to be tested against future queries. To ensure we do not execute this case many times, we require that $A'$ is large and sparse even after removing the previously extracted submatrices from consideration. (Under our parameter settings, we show that $L$ always contains $2^{O(\sqrt{\log n})}$ entries.)
\item The only queries $(u,v)$ remaining are those whose submatrices $A'$ have few entries which are not in previously extracted submatrices. Then, it suffices to find those entries of $A'$ which do not appear among the $(u',v')$ queries saved in $L$, and check to see if any entry is 1. We show how this reduces to the problem of listing orthogonal vectors in a set of $n$ Boolean vectors of dimension equal to $2^{O(\sqrt{\log n})}$.
\end{enumerate}
That is, we ultimately reduce the OMV problem to that of listing $n^2/2^{\Theta(\sqrt{\log n})}$ orthogonal vectors in $2^{\Theta(\sqrt{\log n})}$ dimension.  This listing problem is in turn solved via fast algebraic matrix-matrix multiplication, tweaking a recent algorithm of Abboud, Williams, and Yu~\cite{abboud:batched}. Therefore, we are actually reducing matrix-vector multiplications to smaller matrix-matrix multiplications. This counters the intuition in Henzinger \emph{et al.}~\cite{HenzingerKNS15} that the OMV conjecture addresses the difficulty of developing faster ``combinatorial'' matrix multiplication algorithms: we show how Strassen-like matrix multiplication algorithms can be useful for matrix-vector computations.

\section{Faster Online Matrix-Vector Multiplication}

In this section, we prove:

\begin{reminder}{Theorem~\ref{OMV}} With {\bf no} preprocessing of the matrix $A \in \{0,1\}^{n \times n}$, and for any sequence of $t=2^{\omega(\sqrt{\log n})}$ vectors $v_1,\ldots,v_t \in \{0,1\}^n$, online matrix-vector multiplication of $A$ and $v_i$ over the Boolean semiring can be performed in $n^2/2^{\Omega(\sqrt{\log n})}$ amortized time per query, with a randomized algorithm that succeeds whp. 
\end{reminder}

First, we make use of the following reduction in Henzinger \emph{et al.} \cite{HenzingerKNS15}, which also appears in another form in \cite{Williams2010subcubic}. Define the \emph{online vector-Matrix-vector} problem to be: given an $n \times n$ matrix $A$, and later given online queries $u,v \in \{0,1\}^n$, compute the bit $u^T A v$.

\begin{theorem}[\cite{HenzingerKNS15,Williams2010subcubic}]\label{thm:mvtoumv} Suppose online Boolean vector-Matrix-vector computations can be done in amortized $O(n^2/f(n))$ time per query over a sequence of $t$ queries. Then Boolean online matrix-vector multiplication can be done in amortized $\tilde{O}(n^2/\sqrt{f(n)})$ time per query over a sequence of $t$ queries.
\end{theorem}

Note that the original proof of Theorem~\ref{thm:mvtoumv} does not talk about preserving amortization over $t$ queries. Examining the proof in \cite{HenzingerKNS15}, one sees that a single query for a matrix-vector product $Mv$ reduces to $n$ vector-matrix-vector queries on $n$ different data structures, each representing a $\sqrt{n} \times \sqrt{n}$ submatrix of $M$. For each data structure where the returned answer is $1$, we possibly ask more queries. This means that over a sequence of $t$ queries, all vector-Matrix-vector data structures used in the reduction also processed at least $t$ queries, and thus the amortization kicks in.

We are ready to describe how to solve vector-Matrix-vector multiplication.
Let $M \in \{0,1\}^{n \times n}$ be our given matrix. We want to preprocess $M$ to support queries of the following form: we receive $n$ pairs of vectors $(u_1,v_1),\ldots,(u_n,v_n) \in (\{0,1\}^n)^2$, and must determine $u_i^T M v_i$ before seeing the next pair. 

Given an $n \times n$ matrix $M$, and subsets $U, V \subseteq [n]$, we define $M[U \times V]$ to be the $|U| \times |V|$ submatrix of $M$ restricted to those rows of $M$ with index in $U$, and columns of $M$ with index in $V$. Given $(U,V)$, our goal is to check if $M[U \times V]$ is the all-zeroes matrix. 

\paragraph{Preprocessing.} We initialize a set $C = [n] \times [n]$ and an empty list $L$. The list $L$ will contain triples $(U_k,V_k,S_k)$, where $U_k \subseteq [n]$, $V_k \subseteq [n]$ and $S_k \subseteq [n] \times [n]$, which contains all pairs $(i,j)$ with $i \in U_k$, $j \in V_k$ such that $M(i,j)=1$. Intuitively, the $(U_k,V_k)$ represent \lq\lq{}bad\rq\rq{} queries from the past that we took $O(n^2)$ time to answer. Let $Y, Z > 0$ be integer parameters in the following, where we require $Z \leq n/\log n$. We will maintain the invariants:
 \begin{itemize}
 \item $C$ contains all pairs $(i,j)$ that appear in no $U_k \times V_k$ of $L$. (Informally, $C$ contains ``unseen'' pairs.) It will be helpful to represent $C$ as an $n\times n$ Boolean indicator matrix $D$, where $D(i,j)=1 \iff (i,j) \in C$. 
 \item $|S_k| \leq O(n^2 \log n)/Y$, for all $k=1,\ldots,|L|$. 
 \item $|L| \leq Z$.
 \end{itemize}
Setting up $C$ and $L$ takes $O(n^2)$ time to prepare. (The theorem statement says ``{\bf no} preprocessing'' because we can simply do this $O(n^2)$-time step in the first query.)
 
\paragraph{Query answering.} For each query $(U,V) \subseteq [n] \times [n]$:

\begin{enumerate}
\item {\bf Check for small submatrix.} If $|U|\cdot |V| < n^2/Z$, then try all $i \in U$ and $j \in V$; if $M(i,j)=1$, return $1$.

\item {\bf Check for dense submatrix.} Sample $Y$ uniform random pairs $(i,j) \in U \times V$. If $M(i,j)=1$ for any of these pairs, return $1$. [Otherwise, with high probability, the submatrix $M[U \times V]$ has at most $c (n^2 \log n)/Y$ ones, for some constant $c > 0$.]

\item {\bf Check among pairs seen before, in sparse submatrices.} For all triples $(U_k,V_k,S_k)$ in $L$, and all pairs $(i,j)\in S_k$, if $(i,j) \in U \times V$ then return $1$. 

\item {\bf Estimate the number of pairs in $M[U \times V]$ that have \emph{not} been seen.} Let $R$ be a sample of $n^2/Z$ uniform random entries from $C$. Compute from $R$ an estimate $B$ of the number
$Q := |(U \times V) \cap C|$. In particular, we compute the fraction $\alpha$ of samples from $R$ that lie in $U \times V$. We let our estimate be $B := \alpha |C|$. Clearly the expectation of $B$ is $Q$. Observe that if $Q > 4n^2/Z$, then the expected number of samples that lie in $U \times V$ is $Q|R|/|C| \geq 4n^2/Z^2$. Since we require $Z \leq n/\log n$, we get from a Chernoff bound that $B$ is at least $2n^2/Z$ with high probability. Conversely, if $Q < n^2/Z$, then with high probability, we have $B < 2n^2/Z$.

\item {\bf (a) If estimate is high, brute force.} If $B > 2 n^2/Z$, then do the following in $O(n^2)$ time: 
\begin{compactitem}
\item Compute the answer to the query $(U,V)$. 
\item Determine the set $S = \{(i,j) \in U \times V \mid M(i,j) = 1\}$; note $|S| \leq (c n^2 \log n)/Y$ whp. 
\item Determine the actual quantity $Q = |(U \times V) \cap C|$; If it turns out that $Q < n^2/Z$ or $|S| > (c n^2 \log n)/Y$, we immediately return the answer to the query $(U,V)$.
\item Otherwise, add the triple $(U,V,S)$ to $L$.
\item Remove all $(i,j) \in U \times V$ from $C$, by zeroing out $(i,j)$ entries of $D$. 
\item Return the answer to the query $(U,V)$.
\end{compactitem}

\item {\bf (b) If estimate is low, list the few unseen.} Otherwise, $B \leq 2n^2/Z$. Then with high probability, we only have to check $Q = |(U \times V) \cap C| \leq 4n^2/Z$ entries of $M$ to determine the query answer. Now we need to find the set of pairs
\[ W := (U \times V) \cap C.\]
Provided we can find $W$, we can return the answer to the query by simply checking if there is an $(i,j) \in W$ such that $M(i,j)=1$; if so, return $1$ and if not, return $0$.

\end{enumerate}

Let us estimate the running time based on the query algorithm described so far. Let $T(n,Z)$ be the running time of an algorithm for finding $W$.  Step 1 takes $O(n^2/Z)$ time. For step 2, sampling $Y$ entries and testing for containment takes $O(Y)$ time. Observe that each time we add a triple to $L$ (in step 5), we \begin{itemize}
\item take $O(n^2)$ time to compute the triple,
\item and reduce the size of $C$ by at least $n^2/Z$.
\end{itemize}
Therefore, we add a triple to $L$ at most $Z$ times over all $n$ queries. The total cost of step 5 over all queries is therefore $O(n^2 \cdot Z)$ plus an expected $o(1)$ for the cost of entering step 5 but not adding a triple to $L$. For each query $(U,V)$, we also check all $(i,j)$ entries in all sets $S_k$ in $L$. Since each $S_k$ has $O((n^2 \log n)/Y)$ pairs, and there are only $Z$ sets in $L$, step 3 costs $O((Z n^2 \log n)/Y)$ per query. Step 4 takes $O(n^2/Z)$ time. Since step 6 is assumed to take $T(n, Z)$ time, the overall running time over $q$ queries is thus
\[O\left(q\cdot\left(Y + \frac{Zn^2 \log n}{Y} + \frac{n^2}{Z} + T(n,Z)\right) + n^2 Z\right).\] 

Setting $Y = n^{1.5}$, the bound becomes $\tilde{O}\left(q n^{1.5} + q Z n^{.5} + \frac{q n^2}{Z} + q\cdot T(n,Z) + n^2 Z\right)$. For $q = Z^{\omega(1)}$ and $Z = n^{o(1)}$, the amortized running time is then upper bounded by $\tilde{O}(T(n,Z) + n^2/Z)$.

Now we give an algorithm for reporting $W$, by reduction to listing orthogonal vectors. For every $i \in [n]$, define a vector $u_i \in \{0,1\}^{|L|}$ such that \[u_i[k] = 1 ~\iff~ i \in U_k.\] Similarly for each $j$, define a vector $v_j$ such that $v_j[k] = 1$ $\iff$ $j \in V_k$. 

Recalling our matrix $D$, note that $D(i,j) = 1$ $\iff$ $(i,j) \in C$ $\iff$ $\langle u_i,v_j \rangle = 0$. 
Therefore $(i,j) \in (U \times V) \cap C$ $\iff$ $\left((i,j) \in U \times V \text{ and } D(i,j) = 1\right)$. 
We have reduced the problem of producing $W$ to the following listing problem: 

\begin{quote}

Given two sets of vectors $U \subseteq \{0,1\}^d$, $V \subseteq \{0,1\}^d$ and a look-up table $D$ supporting $O(1)$ time access to $\langle u, v \rangle$ for any pair $u \in U$ and $v \in V$, report all orthogonal pairs in $U \times V$. 
\end{quote}

Note that in our case, we may further assume that
the number of such pairs is at most $K = O(n^2/d)$ with high probability. 
If there are $K$ orthogonal pairs to report, and we get a running time of $T(n,d)=f(n,d)+g(n,d)\cdot K$ as a function of $n$ and $d$, the amortized running time becomes \begin{eqnarray*}
\tilde{O}(f(n,Z)+g(n,Z) \cdot K + n^2/Z)
\end{eqnarray*}

Set $Z=d=2^{\delta \sqrt{\lg n}}$ for some $\delta > 0$. For the OV reporting problem, we partition the vectors into $(n/s)^2$ subproblems on sets of $s = 2^{\eps \delta \sqrt{\lg n}}$ vectors, for some sufficiently small $\eps \in (0,1)$. By a minor modification of the orthogonal vectors algorithm of Abboud, Williams, and Yu~\cite{abboud:batched}, we can report for all $(n/s)^2$ subproblems whether there is an orthogonal pair, in $\tilde{O}((n/s)^2)$ time in total, when $\delta, \eps > 0$ are small enough. 

In particular, their orthogonal vectors algorithm divides the set of $n$ vectors into about $n/s$ groups of $s$ vectors each, and determines for \emph{all} $(n/s)^2$ pairs of groups if there is an orthogonal pair among the $2s$ vectors in the pair of groups. For each subproblem on $2s$ vectors that is reported to contain an orthogonal pair, we  make $O(1)$-time queries into the matrix $D$ to report in $O(s^2)$ time all orthogonal pairs in the group (they are the $1$-entries of $D$ in the appropriate $O(s^2)$-size submatrix).

The running time for the reporting problem is then $\tilde{O}((n/s)^2+K\cdot s^2) = n^2/2^{\Theta(\sqrt{\lg n})}$, and the total running time over $q \geq 2^{\omega(\sqrt{\lg n})}$ queries is $q \cdot n^2/2^{\Theta(\sqrt{\lg n})}$. This completes the OMV algorithm.

\subsection{Applications}
\label{sec:applications}

\paragraph{Graph Algorithms.} There are many natural applications of our OMV data structure to graph algorithms. Here we give a few.

\begin{reminder}{Corollary~\ref{independent}}
For every graph $G = (V,E)$, after $O(n^2)$-time preprocessing there is a data structure such that, for every subset $S \subseteq V$, we can answer online whether $S$ is independent, dominating, or a vertex cover in $G$, in $n^2/2^{\Omega(\sqrt{\log n})}$ amortized time over $2^{\omega(\sqrt{\log n})}$ subset queries.
\end{reminder}

\begin{proof} Let $A$ be the adjacency matrix of $G$. To determine whether $S$ is an independent set, we simply take the column vector $v_S$ which is $1$ in precisely those rows corresponding to vertices in $S$. Then $S$ is independent if and only if $A v_S$ and $v_S$ do not share a $1$ in any coordinate. Note $S$ is independent if and only if $V-S$ is a vertex cover. Finally, $S$ is dominating if and only if $(A v_S) \vee v_S$ is the all-ones vector.
\end{proof} 

\begin{corollary} For every graph $G = (V,E)$, after $O(n^2)$-time preprocessing, there is a data structure such that, for every node $v \in V$, we can answer online whether $v$ is in a triangle of $G$, in $n^2/2^{\Omega(\sqrt{\log n})}$ amortized time over $2^{\omega(\sqrt{\log n})}$ vertex queries.
\end{corollary}

\begin{proof} Node $v$ is in a triangle if and only if the neighborhood of $v$ is not an independent set. We can compute the indicator vector $w$ for the neighborhood of $v$ in $O(n)$ time, then determine if $w$ corresponds to an independent set using Corollary~\ref{independent}.
\end{proof}

\paragraph{Partial Matches.} Recall in the partial match retrieval problem, we are given $x_1,\ldots,x_n \in \{\Sigma \cup \{\star \}\}^m$ with $m \leq n$, and wish to preprocess them to answer ``partial match'' queries $q \in (\Sigma \cup \{\star\})^m$, where we say $q$ matches $x_i$ if for all $j=1,\ldots,m$, $q[j] = x_i[j]$ whenever $q[j] \neq \star$ and $x_i[j] \neq \star$. A query answer is the vector $v \in \{0,1\}^n$ such that $v[i] = 1$ if and only if $q$ matches $x_i$. For this problem, we show:

\begin{reminder}{Theorem~\ref{thm:partialmatch}} 
Let $\Sigma = \{\sigma_1,\ldots,\sigma_k\}$ and let $\star$ be an element such that $\star \notin \Sigma$. For any set of strings $x_1,\ldots,x_n \in \{\Sigma \cup \{\star \}\}^m$ with $m \leq n$, after $\tilde{O}(nm)$-time preprocessing, there is an $O(nm)$ space data structure such that, for every query string $q \in (\Sigma \cup \{\star\})^m$, we can answer online whether $q$ matches $x_i$, for every $i=1,\ldots, n$, in $(nm \log k)/2^{\Omega(\sqrt{\log m})}$ amortized time over $2^{\omega(\sqrt{\log m})}$ queries.
\end{reminder}

\begin{proof} First, for simplicity we assume that $n = m$; at the end, we show how to handle the case $m \leq n$. 

Build an $n \times n$ matrix $A$ over $(\Sigma \cup \{\star\})^{n \times n}$ such that $A[i,j] = x_i[j]$. We efficiently ``Booleanize'' the matrix $A$, as follows. Let $S_1,T_1,\ldots,S_k,T_k \subseteq[2\log k]$ be a collection of subsets such that for all $i$, $|S_i \cap T_i| = \emptyset$, yet for all $i \neq j$, $|S_i \cap T_j| \neq \emptyset$. Such a collection exists, by simply taking (for example) $S_i$ to be the $i$th subset of $[2\log k]$ having exactly $\log k$ elements (in some ordering on sets), and taking $T_i$ to be the complement of $S_i$. (The construction works because there are $\binom{2\log k}{\log k} > k$ such subsets, and because any two subsets over $[2\log k]$ with $\log k$ elements must intersect, unless they are complements of each other.) Extend the matrix $A$ to an $n \times (2n\log k)$ Boolean matrix $B$, by replacing every occurrence of $\sigma_i$ with the $(2 \log k)$-dimensional row vector corresponding to $S_i$, and every occurrence of $\star$ with the $(2 \log k)$-dimensional row vector which is all-zeroes. 

Now when a query vector $q \in (\Sigma \cup \{\star\})^n$ is received,  convert $q$ into a Boolean (column) vector $v$ by replacing each occurrence of $\sigma_i$ with the $(2\log k)$-dimensional (column) vector corresponding to $T_i$, and every occurrence of $\star$ by the $(2 \log k)$-dimensional (column) vector which is all-zeroes. Compute $A v$ using the OMV algorithm. For all $i=1,\ldots,n$, we observe that $q$ matches $x_i$ if and only if the $i$th row of $B$ is \emph{orthogonal} to $v$. The two vectors are orthogonal if and only if for all $j=1,\ldots,n$, either the $i$th row of $B$ contains the all-zero vector in entries $(j-1)(2\log k) + 1,\ldots,j(2\log k)$, or in those entries $B$ contains the indicator vector for a set $S_{\ell}$ and correspondingly $v$ contains either $\star$ or a set $T_{\ell'}$ such that $S_{\ell} \cap T_{\ell'} = \emptyset$, i.e. $x_i$ and $q$ match in the $j$th symbol. That is, the two vectors are orthogonal if and only if $q$ matches $x_i$. Therefore, $Av$ reports for all $i=1,\ldots,n$ whether $q$ matches $x_i$ or not. 

Let us discuss the running time. By splitting the $n \times (2n\log k)$ matrix $B$ into $2\log k$ matrices of dimension $n \times n$, we can answer queries in $O(n^2 \log k)/2^{\Omega(\sqrt{\log n})}$ time via Theorem~\ref{OMV}, by computing the $n$-dimensional output vectors for each of the $2 \log k$ matrices, and taking their component-wise OR. 

For $m \leq n$, the matrix $B$ will have dimension $n \times (2m\log k)$. Splitting $B$ into at most $n/m$ matrices of dimension $m \times (2m\log k)$, we can answer the queries for one of these matrices in $O(m^2 \log k)/2^{\Omega(\sqrt{\log m})}$ time, returning $m$-dimensional vectors for each one. The final query answer is then the concatenation of these $n/m$ vectors. Since this query answering happens for $n/m$ matrices, the final running time is then $O(nm \log k)/2^{\Omega(\sqrt{\log m})}$.

For the space usage, recall our OMV data structure simply stores the matrix $B$ (which is $O(nm)$ words of space) plus another $O(nm)$ additive space usage for the structures $L$, $C$ and $D$.
\end{proof}

\paragraph{2-CNF Evaluation.} Our last application was first observed by Avrim Blum (personal communication). Namely, we can evaluate a 2-CNF formula on a variable assignment faster than plugging in the assignment to all possible clauses.

\begin{corollary} For every 2-CNF Boolean formula $F$ on $n$ variables, after $O(n^2)$-time preprocessing, there is a data structure such that, for every assignment $A : [n] \rightarrow \{0,1\}$ to the variables of $F$, we can answer online the value of $F(A)$ in $n^2/2^{\Omega(\sqrt{\log n})}$ amortized time over $2^{\omega(\sqrt{\log n})}$ queries.
\end{corollary}

\begin{proof} Given $F$ on variables $x_1,\ldots,x_n$, build a graph on $2n$ nodes, which has a node for every possible literal $x_i$ and $\neg x_i$. For every clause $(\ell_i \vee \ell_j)$ in $F$, put an edge between the literals $\neg \ell_i$ and $\neg \ell_j$ in the graph. Now given a variable assignment $A : [n] \rightarrow \{0,1\}$, observe that the set $S = \{x \mid A(x)=1\} \cup \{\neg x \mid A(x) = 0\}$ is independent if and only if $F(A) = 1$. Now we appeal to Corollary~\ref{independent}.
\end{proof}

\paragraph{Extension to Worst-Case Query Time.} We also note that the amortized time bound can be converted into a worst-case one, at the cost of slightly-exponential preprocessing time (but linear space):

\begin{theorem} \label{worst-case} Let $\eps \in (0,1]$. After $\exp(O(n^{\eps}))$-time preprocessing time of any matrix $A \in \{0,1\}^{n \times n}$, there is a data structure using $O(n^2)$ bits of space, such that given a vector $v \in \{0,1\}^n$, we can compute $Av$ over the Boolean semiring in $n^2/2^{\Omega(\sqrt{\eps \log n})}$ {\bf worst-case} time. The algorithm is randomized and succeeds whp. 
\end{theorem}

\begin{proof} First consider the case $\eps = 1$. We solve the vector-Matrix-vector problem with the guarantees of the theorem and use the reduction of Henzinger \emph{et al.} \cite{HenzingerKNS15} to get the same bounds for Matrix-vector multiplication. Our preprocessing phase for the vector-Matrix-vector problem will do the following: while there is \emph{any} possible query $(U, V)$ that would cause our amortized algorithm to enter case 5 (the brute-force step) and add a triple to $L$, we execute step 5 for that query $(U,V)$. This adds the triple $(U,V,S)$ to $L$. When there are no such queries left, the preprocessing phase stops. Note that finding such an $(U,V)$ may require an $\exp(O(n))$-time exhaustive search over all possible pairs $(U,V)$.

Since step 5 is the only step of the query algorithm that costs more than $n^2/2^{\Omega(\sqrt{\log n})}$ time. and there are no queries left that can enter step 5 and add a triple to $L$, we conclude that any query we could ask the data structure, can only enter step 5 with probability $1/\poly(n)$. Thus if we simply guess the answer to a query if it is about to enter step 5, we conclude that the resulting data structure has worst case query time $n^2/2^{\Omega(\sqrt{\log n})}$ and is correct whp.

The preprocessing time can be reduced to $\exp(O(n^{\eps}))$ for any desired $\eps > 0$, by partitioning the matrix into $n^{2-2\eps}$ blocks of dimensions $n^{\eps} \times n^{\eps}$ each, preprocessing each block separately in $\exp(O(n^{\eps}))$. Given a vector-Matrix-vector query, we simply run the vector-Matrix-vector algorithm over all $n^{2-2\eps}$ blocks, each run taking $n^{2\eps}/2^{\Omega(\sqrt{\eps \log n})}$ time. The overall worst-case query time is $n^2/2^{\Omega(\sqrt{\eps \log n})}$.
\end{proof}

\subsection{OMV in the Cell Probe Model}

In this section we prove Theorem~\ref{OMVCellProbe}, showing how to efficiently solve OMV in the cell probe model. We start by giving an efficient cell probe data structure for Boolean vector-Matrix-vector multiplication, and then apply the reduction of Henzinger \emph{et al.}~\cite{HenzingerKNS15}. For vector-Matrix-vector multiplication, we show the following:

\begin{theorem}
\label{uMvCellProbe}
There is a cell probe data structure that given a matrix $A \in \{0,1\}^{n \times n}$, preprocesses the matrix into an $O(n^2)$ space data structure, such that for a pair of query vectors $u,v \in \{0,1\}^n$, the data structure computes $u^TAv$ in worst case $O(n^{3/2}/\sqrt{w})$ probes. Here $w$ is the word size.
\end{theorem}

Our data structure is even succinct in the sense that it stores $A$ as it is, plus an additional $O(n^{3/2}\sqrt{w})$ redundant bits. Before proving Theorem~\ref{uMvCellProbe}, we show that it implies Theorem~\ref{OMVCellProbe}. The reduction in~\cite{HenzingerKNS15} takes an input matrix $A$ for OMV and blocks it into $n$ submatrices of $\sqrt{n} \times \sqrt{n}$ entries each. A vector-Matrix-vector data structure is then implemented on each of these submatrices. On a query vector $v \in \{0,1\}^n$ for OMV, one asks $O(n)$ vector-Matrix-vector queries on these data structures. It follows that by plugging in the data structure from Theorem~\ref{uMvCellProbe}, we obtain an OMV data structure that simply stores the matrix $A$ plus an additional $O(n n^{3/4}\sqrt{w})=O(n^{7/4} \sqrt{w})$ redundant bits and that answers queries in $O(n^{7/4}/\sqrt{w})$ cell probes.

We are ready to describe our cell probe data structure for vector-Matrix-vector multiplication. Our data structure is in fact a simplification of our word-RAM solution.

\paragraph{Preprocessing.}
When given the input matrix $A \in \{0,1\}^{n \times n}$, we construct a list $L$ consisting of pairs $(U,V)$ where $U$ is a subset of rows in $A$ and $V$ is a subset of columns in $A$. Each pair $(U,V)$ in $L$ will have the property that the corresponding submatrix $A[U \times V]$ has only zeros. We construct $L$ as follows: While there exist a pair $(U,V)$ such that $A[U \times V]$ has only zeros and
$$
\left| (U \times V) \setminus \bigcup_{(U',V') \in L} (U' \times V')\right| \geq n^{3/2}/\sqrt{w},
$$
we add the pair $(U,V)$ to $L$. When this terminates, we know that $|L| \leq n^{1/2}\sqrt{w}$. Our data structure simply stores $A$ plus the list $L$. Since each pair in $L$ can be described using $O(n)$ bits (an indicator bit per row and column), the data structure stores $O(n^{3/2} \sqrt{w})$ bits in addition to the input matrix $A$.

\paragraph{Query answering.}
When given a query pair of vectors $u,v \in \{0,1\}^n$, let $U$ denote the subset of rows indexed by $u$ and $V$ the subset of columns indexed by $v$. We must determine whether there is a one in $A[U \times V]$. To do this, we start by reading the entire list $L$. This costs $O(n^{3/2}/\sqrt{w})$ cell probes. Since computation is free, we now compute the set $Q$, where
$$
Q :=  (U \times V) \setminus \bigcup_{(U',V') \in L} (U' \times V').
$$
If $|Q| \geq n^{3/2}/\sqrt{w}$, then we know that the answer to the query is one, since otherwise the pair $(U,V)$ could have been added to $L$ during preprocessing. Otherwise, we have $|Q| < n^{3/2}/\sqrt{w}$. Since all submatrices $A[U' \times V']$ for a pair $(U',V') \in L$ contains only zeroes, we know that any entry $(i,j)$ inside $A[U \times V]$ can only contain a one if it is also in $Q$. We thus explicitly check all such $|Q| = O(n^{3/2}/\sqrt{w})$ entries of $A$. If any of them is a one, we return the answer one. Otherwise, we return the answer zero.

\paragraph{Discussion.}
Our cell probe algorithm above heavily exploits the non-uniformity of the cell probe model when computing the set $Q$, and deciding what to do based on $Q$. In our word-RAM algorithm for OMV, we had to compute $Q$ using Orthogonal Vectors, which significantly slows down the data structure compared to the cell probe version.

As another remark, recall that the OMV conjecture talks about the total time for preprocessing $A$ and then answering $n$ queries, one at a time. Observe that our above cell probe data structure also gives an efficient data structure for this setup, since preprocessing $A$ simply costs $O(n^2/w)$ cell probes for reading the entire matrix and then simply writing out the data structure with worst case query time of $O(n^{7/4}/\sqrt{w})$. This is because the time spent finding pairs $(U,V)$ is solely computational, i.e., once all of $A$ has been read, we know what data structure to build without any further cell probes.

\section{Conclusion}

We have shown how ``off-line'' algebraic algorithms, in particular fast matrix multiplications, can be applied to solve basic online query problems such as matrix-vector multiplication. Our approach uses both combinatorial and algebraic ideas to achieve these results.

The OMV algorithm given is effectively a reduction from Boolean matrix-vector product to Boolean matrix-matrix product. In particular, it shows that the OMV problem can (in some sense) be reduced to the Orthogonal Vectors problem, which has been studied in similar contexts. Is there a deeper relationship between these two problems? In particular, our proof shows that if one can list $n^2/Z$ pairs of orthogonal vectors in $Z \geq n^{\eps}$ dimensions in $n^{2-\gamma}$ time for some $\gamma,\eps >0$, then the OMV conjecture is false. 

It would be tempting to conclude from the above implication that an $n^{2-\gamma}$ time algorithm for \emph{detecting} an orthogonal pair in $n^{\eps}$ dimensions breaks the OMV conjecture. This would show that the OMV conjecture implies that Orthogonal Vectors in $n^{\eps}$ dimensions is not in truly subquadratic time. Such a connection would establish OMV hardness for problems that have previously been proved to be SETH-hard via Orthogonal Vectors. Indeed, a standard approach for turning a detection algorithm into a listing algorithm almost seems to establish this connection: Consider listing orthogonal pairs of vectors, where we are given two sets $A$ and $B$, both of $n$ vectors in $n^{\eps}$ dimensions. The goal is to list all orthogonal pairs $(a,b)$ with $a \in A$ and $b \in B$. Now partition both $A$ and $B$ into two halves $A_0, A_1, B_0$ and $B_1$. Run a detection algorithm for every pair of sets $(A_i, B_j)$ and recurse on the sets where the detection algorithm says there is an orthogonal pair. Such an approach has previously been successful in obtaining a listing algorithm from a detection algorithm. The caveat here is that, while the number of vectors halves each time we recurse, the dimensionality of the vectors stays the same. Eventually, the dimensionality is so large compared to the subproblem size, that the detection algorithm no longer runs in $n^{2-\gamma}$ time, where $n$ is the subproblem size. 

Finally, we showed that our solution can be implemented in a very efficient manner in the cell probe model. The cell probe version of our data structure answers OMV queries in strongly subquadratic time. This rules out any chance of giving an unconditional proof of the OMV conjecture in the cell probe model. It would be very exciting if this improvement could be implemented in the word-RAM, and thereby refute the OMV conjecture.
\bibliographystyle{alpha}
\bibliography{papers}

\end{document}